\documentclass[%
 reprint,
 amsmath,amssymb,
 aps,
]{revtex4-2}

\usepackage{graphicx}% Include figure files
\usepackage{dcolumn}% Align table columns on decimal point
\usepackage{bm}% bold math

\usepackage{amsmath,amssymb,amsthm,easybmat,mathtools}
\usepackage{braket}
\usepackage{graphicx}% Required for inserting images
\usepackage{bbold}
\usepackage{soul}

\newtheorem{corollary}{Corollary}

\newtheorem{lemma}{Lemma}
\newtheorem{proposition}{Proposition}

\newtheorem{theorem}{Theorem}
\theoremstyle{definition}
\newtheorem{remark}{Remark}
\newtheorem{example}{Example}

\newcommand{\mc}{\mathcal}
\newcommand{\tr}{\text{Tr}}
\newcommand{\mbb}{\mathbb}
\newcommand{\mbf}{\mathbf}
\newcommand{\id}{\text{id}}
\newcommand{\cptp}{\text{CPTP}}
\newcommand{\op}[2]{|#1\rangle\langle #2|}
\newcommand{\ip}[2]{\langle #1| #2 \rangle}

\usepackage{color}
\definecolor{cool_green}{rgb}{0.0, 0.5, 0.0}

\begin{document}

\preprint{APS/123-QED}

\title{Quantum Channel Masking}% Force line breaks with \\

\author{Anna Honeycutt }
\email{annagh2@illinois.edu}
\affiliation{Department of Physics, \\University of Illinois Urbana-Champaign, Urbana, IL, USA}

\author{Hailey Murray}
\email{hm649@cornell.edu}
\affiliation{Department of Physics,\\
Embry-Riddle Aeronautical University, Prescott, AZ, USA}
\affiliation{School of Applied and Engineering Physics,\\
Cornell University, Ithaca, NY, USA}

\author{Eric Chitambar}
\email{echitamb@illinois.edu}
\affiliation{Department of Electrical and Computer Engineering, 
Coordinated Science Laboratory, \\
University of Illinois Urbana-Champaign, Urbana, IL, USA}

\date{\today}

\begin{abstract}
Quantum masking is a special type of secret sharing in which some information gets reversibly distributed into a multipartite system, leaving the original information inaccessible to each subsystem. This paper proposes a dynamical extension of quantum masking to the level of quantum channels.  In channel masking, the identity of a channel becomes locally hidden but still globally accessible after its output is sent through a bipartite broadcasting channel.  We first characterize all families of $d$-dimensional unitaries that can be isometrically masked, a condition that holds even in the presence of depolarizing noise.  For the case of qubits, we identify which families of Pauli channels can be masked, and we prove that a qubit channel can be masked against the identity if and only if it is unital and has a pure-state fixed point.  Masking against the identity describes a scenario in which channel noise becomes completely delocalized through a broadcast map and undetectable through subsystem dynamics alone.

\end{abstract}

\maketitle

%\section{\label{sec:level1}Introduction/Background}

There are several no-go theorems of great importance to quantum information processing, such as the no-cloning \cite{Wootters1982} and the no-broadcasting theorems \cite{Barnum1996,Kalev2008}. Additionally, the no-deleting theorem states that we cannot delete unknown quantum information in a closed system \cite{Pati2000}, and due to the no-hiding theorem, quantum information that is lost in one subsystem of a closed composite system must be recoverable from the other subsystem \cite{Braunstein2007}. These no-go theorems have had significant applications in error correction \cite{Kretschmann-2008a}, quantum key distribution \cite{Gisin2002} and secret sharing \cite{Karlsson1999, Hillery1999}.

%From the existing no-go theorems we obtain a notion of conservation of quantum information. 
 %If we cannot completely delete or hide quantum information, can it perhaps be obscured from individual subsystems? 
 As a modification to the original problem of hiding quantum information \cite{Braunstein2007}, Modi et al. proposed the task of quantum masking \cite{Modi2018}. This involves reversibly splitting states from a given set into two parts such that the identity of the original state cannot be determined by examining either of the parts individually (see Fig. \ref{Fig:state_masking}).  While the no-hiding theorem implies this is impossible for arbitrary states drawn from a Hilbert space with a pure ancilla system, it was shown that masking is possible for restricted sets of states \cite{Modi2018}.  In fact, maskable sets of states can be much richer than cloneable sets, the latter consisting of just orthonormal pure states.
%proposed an additional no-go theorem: the no-masking theorem. The no-masking theorem argues that it is impossible to mask all arbitrary quantum states by using the same masking operator. The authors show that there exists sets of quantum states such that by reversibly splitting the information into a bipartite system, the original state cannot be distinguished by accessing each subsystem.
For qubits, any set of states whose Bloch vectors fall on a single disk in the Bloch sphere can be masked \cite{Liang2019, Ding2020}, and a similar geometrical structure appears to hold for the sets of higher-dimensional maskable states  \cite{Ding2020}.  Additionally, it has been shown that all qubit states can be masked if we move beyond bipartite splitting and consider multipartite masking \cite{Li2018}.  This is a direct consequence of the fact that every quantum error correcting code accomplishes the task of quantum masking, and connections between the two have been studied \cite{Li2018, Han-2020a}.  In particular, multipartite masking enables the possibility of quantum secret sharing \cite{Karlsson1999,Hillery1999,Zukowski1998}, where quantum information is distributed among several parties such that the original information is inaccessible without collaboration.  Masking also has applications in other cryptographic tasks such as quantum bit commitment \cite{Mayers1997, Modi2018}. %The authors of Ref. \cite{Modi2018} propose a no-qubit commitment protocol in which Alice commits to a qubit, but can always cheat.

\begin{figure}[t]
\includegraphics[width=8cm]{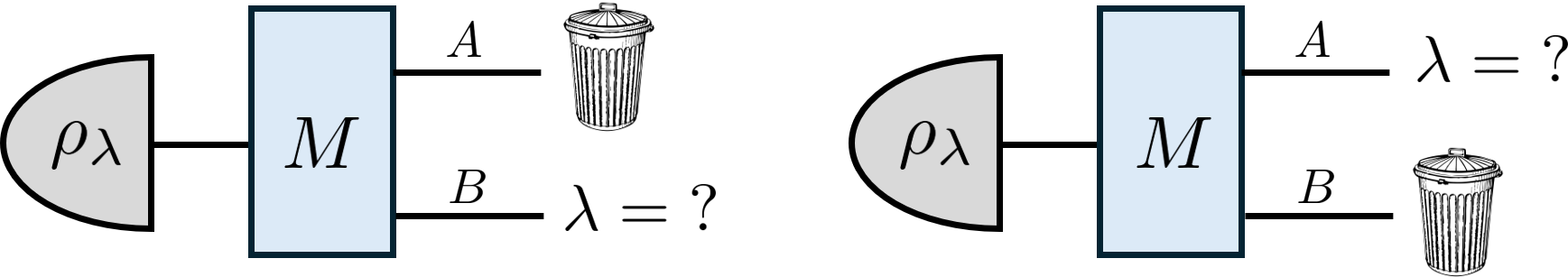}
\caption{In state masking, a set of states $\mc{S}=\{\rho_\lambda\}_\lambda$ is masked by the isometry $M$ such that the reduced state outputs are independent of $\lambda$. Here, the trash can is used to represent a discarding
of the subsystem, which mathematically corresponds to a partial trace.}
\label{Fig:state_masking}
\end{figure}

%%%%%
%\todo{[This original paragraph talks about crypto applications for state masking.  Let's modify it to better motivate the crypto applications of channel masking.]}  
%This masking protocol has applications in quantum bit commitment \cite{Mayers1997}. The authors of Ref. \cite{Modi2018} propose a no-qubit commitment protocol in which Alice commits to a qubit, but can always cheat. The no-masking theorem has significant implications for quantum secret sharing \cite{Karlsson1999,Hillery1999,Zukowski1998} as well, where quantum information is distributed among several parties such that the original information is inaccessible without collaboration. This information may then be encoded in a set of maskable states.

\begin{figure}[t]
    \centering
    \includegraphics[width=8cm]{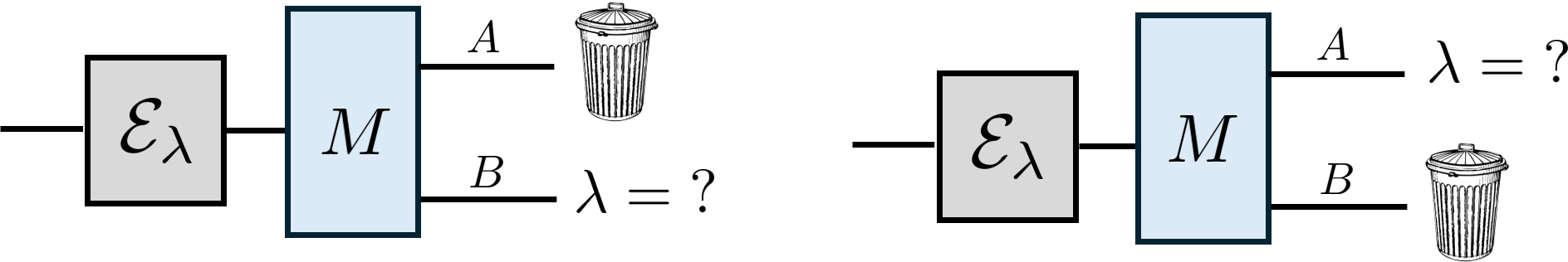}
    \caption{In channel masking, a set of channels or gates $\mc{S}=\{\mc{E}_\lambda\}_\lambda$ is masked by the isometry $M$ such that the reduced channels are independent of $\lambda$.}
\label{Fig:channel_masking}
\end{figure}

Here, we propose an extension of quantum masking to sets of quantum channels or quantum gates.  As depicted in Fig. \ref{Fig:channel_masking}, this is a generalization of state masking in which some reversible splitting is performed as post-processing to a random quantum channel $\mc{E}_\lambda$ chosen from some set $\mc{S}=\{\mc{E}_\lambda\}_\lambda$.  The masking is successful if the two reduced channels are independent of $\lambda$, thereby hiding the ``which channel'' information from local recovery; only by examining the global channel $\mc{U}_M\circ\mc{E}_\lambda(\cdot)\equiv M[\mc{E}(\cdot)]M^\dagger$ can the value $\lambda$ be perfectly recovered.  This problem is similar in spirit but conceptually distinct from the well-studied task of hiding the ``which channel'' information of a state-dependent channel between a sender and one receiver \cite{Shannon-1958a, Boche-2016a, Pereg-2021a}.  In that scenario, one attempts to minimize the leakage rate of the state parameter $\lambda$ through channel encoding.  In contrast, the problem considered here involves a single-shot use of the channel or gate, and the parameter $\lambda$ gets masked in the bipartite correlations generated after the channel or gate. 

A special case of this problem considers masking just a single noisy channel $\mc{E}$ from the ideal identity channel, $\text{id}$. In this scenario, the masker serves as a broadcasting map $M:\mc{H}_Q\to \mc{H}_{AB}$ that completely hides the effect of $\mc{E}$ from the subsystems $A$ and $B$, i.e. 
\begin{align}
\label{Eq:identity-mask}
\tr_X [M\rho M^\dagger] = \tr_X[M\mc{E}(\rho) M^\dagger],\quad X\in\{A,B\},\;\forall \rho.
\end{align}
Since the broadcasting can be inverted, the information of whether or not the noisy map is present must be accessible somewhere in the overall system, and in this case, that information lies entirely in the correlations between subsystems $A$ and $B$.

%Current literature on state masking restricts to isometric maskers to model the physical, reversible process \cite{Modi2018,Ding2020,Liang2019,Li2018}. Relaxing this constraint for channel masking allows for the potential of eavesdroppers and dishonest parties. For example, modeling masker $\mc{M}$ in the quantum-classical scheme $\mc{M}(\rho) = \frac{1}{4} \sum_{m,n} ( X^m Z^n \mc{E}(\rho) Z^n X^m )_A \otimes \op{mn}{mn}_B$ allows for reversible masking by distributing the classical key (m,n) to Bob. However, if a dishonest Alice has access to side channels, or collaborates with an eavesdropper, Alice would obtain the classical key to recover the global initial state. \textcolor{red}{cite Buscemi?} We restrict to isometric masking in order to decouple from the environment completely.

In this work, we focus on the masking of qubit channels and sets of qudit (unitary) gates.  As one of our main results, we derive necessary and sufficient conditions for when an arbitrary number of qudit gates can be masked, and we provide an explicit construction of a masker.  For qubit gates, these conditions admit an appealing geometrical interpretation on the Bloch sphere.  We then focus on the masking of qubit channels and characterize all the families of noisy Pauli channels that can be masked.  We further identify all qubit channels whose action can be masked against the identity (i.e. satisfy Eq. \eqref{Eq:identity-mask}).  Finally, we consider the masking of classical channels.  While classical Boolean circuits are unable to mask families of classical channels, we observe that conjugate coding provides a natural construction of a quantum masker of arbitrary classical channels.  Before presenting our results, we provide a more precise statement of our problem and an overview of definitions.  We then proceed to our main results and conclude the paper with a brief discussion on future directions of work.

%This paper is organized as follows. In Section II we provide preliminary definitions, lemmas, and theories concerning quantum channels and hyperdisks that will be applied in later sections. In Section III, we present the main results of the paper. In section IV we refer to specific maskable channels, and discuss the advantage of quantum circuits in the masking of classical channels. Finally, in Section V we summarize our final results and address possible areas of future work.

%\section{\label{sec:level1}Preliminaries}

\medskip

\section{Preliminaries}

We begin by defining the notion of state masking presented in Ref. \cite{Modi2018}.  Let $M:\mc{H}_Q\to\mc{H}_{AB}$ be an isometric mapping (meaning $M^\dagger M=\mbb{1}_Q$) from system $Q$ into a bipartite system $AB$.  The map is said to mask quantum information contained in the set of states $\mc{S}=\{\ket{\psi_\lambda}\}_\lambda \subset \mc{H}_Q$ if the states $\{\ket{\Psi_\lambda}_{AB}:=M\ket{\psi_\lambda}_Q\}_\lambda$ have the same marginal states, i.e. 
    \begin{align}
        \rho_A &= \tr_B\op{\Psi_\lambda}{\Psi_\lambda},
        &\rho_B &= \tr_A\op{\Psi_\lambda}{\Psi_\lambda}
    \end{align}
for all $\lambda$. By definition the marginal states contain no information on the original input state. However, the choice of isometric mapping $M$ allows for the recovery of this information.

We remark that the definition of state masking presented here follows current literature and restricts $M$ to be an isometric mapping \cite{Modi2018,Ding2020,Liang2019,Li2018}.  Of course, one could also consider maskers $\mc{M}$ that are more general completely-positive trace-preserving (CPTP) maps.  However, this more general model of masking allows for the potential of eavesdroppers and dishonest parties.  For example, the full set of qubit states gets masked under the mapping $\op{\psi}{\psi}\mapsto \mc{M}(\op{\psi}{\psi}) = \frac{1}{4} \sum_{m,n} ( X^m Z^n \op{\psi}{\psi} Z^n X^m )_A \otimes \op{mn}{mn}_B$, where $X$ and $Z$ are Pauli operators; one can directly verify that both $\tr_A\mc{M}(\op{\psi}{\psi})$ and $\tr_B\mc{M}(\op{\psi}{\psi})$ are completely mixed for every $\ket{\psi}$.  Yet, in a purified picture, up to a basis change on $E$ the map $\mc{M}$ arises from some dilation $U_{Q\to ABE}=\sum_{m,n}X^mZ^n_ \otimes \ket{mn}_B\otimes \ket{mn}_E$ such that $\mc{M}(\rho)=\tr_{E}(U\rho U^\dagger)$.  If Alice were to have access to system $E$, then she could correct the Pauli $Z^nX^m$ on her system and perfectly obtain the state $\ket{\psi}$.  One might then be motivated to demand the strongest form of masking and require that Alice has no information about the input state even if she had access to any additional side information outside of Bob's system, meaning that the output of the masker should be a bipartite pure state.  Such is the attitude taken in traditional studies of quantum key distribution (QKD) \cite{Lo-1999a}, and it is one we likewise adopt here.  Thus, throughout this work we will assume that all maskers are isometric in form.

The task of bipartite state masking has been completely characterized in the qubit setting \cite{Liang2019,Ding2020}, as initially conjectured in Ref. \cite{Modi2018}. Specifically, a set of maskable qubit states corresponds to a two-dimensional hyperdisk on the Bloch sphere \cite{Ding2020}.  A similar geometric structure is suggested to persist in higher dimensions, as evidenced in the case of qutrits \cite{Ding2020}, though a complete characterization of d-dimensional masking remains to be studied. While universal bipartite masking of arbitrary states is impossible \cite{Modi2018}, it has been shown that universal multipartite masking is possible: all states in $\mbb{C}^d$ may be masked by a multipartite masker with the addition of $2d-1$ systems of dimension $d$ \cite{Li2018}.

We now extend the concept of state masking to quantum channels. A quantum channel $\mc{E}$ is defined as a CPTP map from input system $Q$ to output system $Q'$, and the collection of all such maps we denote as $\cptp(Q\to Q')$. %and can be represented via Stinespring dilation as
%\begin{align}
%    \mc{E_U}^{Q\to Q'}(\rho) := Tr_{R'}[U^{QR\to Q'R'}(\rho \otimes \omega^R)U^{\dag QR\to Q'R'}]
%\end{align}
%for some isometry $U$ and ancilliary state $\omega^R$.
We say that a set of channels $\mc{S} = \{\mc{E}_\lambda\}_\lambda \subset \cptp(Q \to Q')$ is maskable if there exists an isometry $M : \mc{H}_{Q'} \to \mc{H}_{AB}$ such that 
    \begin{align}
    \label{Eq:Channel-masking-defn}
        \alpha(\cdot) &= \tr_B[M \mc{E_\lambda(\cdot)} M^\dag],
        &\beta(\cdot) &= \tr_A[M \mc{E_\lambda(\cdot)} M^\dag],
    \end{align}
where $\alpha(\cdot)$ and $\beta(\cdot)$ are fixed channels independent of $\lambda$. 
The latter conditions mean that the identity of the applied channel $\mc{E}_\lambda$ is hidden from each subsystem, while still being globally recoverable via the isometry $M$ (see Fig. \ref{Fig:channel_masking}). %\todo{[Anna, how would you recover the value $\lambda$ globally?]}. 
A special instance of this problem involves families of unitary gates $\{U_\lambda\}_\lambda$, and we say that a set of gates can be masked if the corresponding set of channels $\{\mc{U}_{\lambda}\}_\lambda$ can be masked, where $\mc{U}_{\lambda}(\cdot):=U_\lambda(\cdot)U_\lambda^\dagger$. 

A useful observation is that applying pre- and post- unitaries to a family of channels does not alter its ability to be masked.  That is, the channels $\{\mc{E}_\lambda\}_\lambda$ can be masked iff $\{\mc{U}_{\text{post}}\circ\mc{E}_\lambda\circ\mc{U}_{\text{pre}}\}_\lambda$ can be masked for any choice of unitaries $U_{\text{pre}}$ and $U_{\text{post}}$ independent of $\lambda$, a fact that can be directly verified from the definition in Eq. \eqref{Eq:Channel-masking-defn}.  These unitary degrees of freedom will be used heavily when studying the masking of qubit channels below.

\medskip

\section{Results}

\subsection{Masking families of unitaries}

We now present our findings on the problem of channel masking.  We begin by restricting our attention to the masking of unitary gates $\{U_\lambda\}_{\lambda}$ that act on some $d$-dimensional quantum system $\mc{H}_Q\cong\mbb{C}^d$.  The following structural lemma pertains to the masking of a single unitary $U$ and the identity $\mbb{1}$.\\

%\textcolor{red}{link to error correcting code [}
%We restrict $M$ to be an isometric mapping to ensure a secure environment as follows. 

%Let $\mc{M} \in \cptp(Q \to AB)$ be a channel defined by a set of Kraus operators $\{K_i\}$ that masks all unitary operators $\{U_\lambda\}$. Consider the Stinespring dilation of $\mc{M}$ defined by the isometry $V: \mc{H}_Q \rightarrow \mc{H}_{AB} \otimes \mc{H}_E$
%\begin{align}
 %   V\ket{\psi} = \sum_i K_i \ket{\psi} \otimes \ket{i}_E
%\end{align}
%After masking the application of a unitary $U_\lambda$ on some arbitrary input state $\rho$, the reduced state on the environment, $\rho_E$, now becomes
%\begin{align}
 %   \tr_{AB}[V U_{\lambda} ( \rho ) U_\lambda^\dag V^\dag]
  %  &= \sum_{i,j} \tr[K_j^\dag K_i U_\lambda(\rho) U_\lambda^\dag] \;  \op{i}{j}_E\notag\\
   % &= \sum_{ij} \tr[K_j^\dag K_i \rho_\lambda] \;  \op{i}{j}_E
%\end{align}
%To eliminate dependence on $\lambda$, this must be some function of $i,j$, constant over all $\rho_\lambda$. Then we demand $K_j^\dag K_i = c_{ij} \mathbb{1}$ for all $i,j$. \textcolor{red}{]}

\begin{proposition}
\label{Prop:identity-unitary-mask}
    Suppose that $M:\mc{H}_{Q}\to\mc{H}_{AB}$ is a masker for the two unitaries $\{\mathbb{1},U\}$.  Let $\ket{e_1}$ and $\ket{e_2}$ be any two eigenstates of $U$ belonging to distinct eigenspaces.  Then $M$ must map $\ket{e_1}$ and $\ket{e_2}$ to locally orthogonal states.  In other words,
    \begin{align}
        \tr_X (M\op{e_1}{e_1}M^\dagger) \perp \tr_X (M\op{e_2}{e_2}M^\dagger) 
    \end{align}
    for $X\in\{A,B\}$, where $\perp$ denotes operators with orthogonal supports.
\end{proposition}
\noindent In the language of \cite{George2023}, the masker $M$ must ``broadcast the orthogonality'' of the eigenstates $\ket{e_1}$ and $\ket{e_2}$.
\begin{proof}
    Consider an arbitrary superposition $\ket{\psi}=\cos\theta\ket{e_1}+\sin\theta e^{i\phi}\ket{e_2}$ such that $U\ket{\psi}=\cos\theta \lambda_1\ket{e_1}+\sin\theta e^{i\phi}\lambda_2\ket{e_2}$, where $\lambda_1$ and $\lambda_2$ are the distinct eigenvalues of $U$ for eigenstates $\ket{e_1}$ and $\ket{e_2}$. Write $\ket{E_i}_{AB}=M\ket{e_i}$.  Then since $|\lambda_1|^2=|\lambda_2|^2=1$, we have
\begin{align}
    \tr_X (M&\op{\psi}{\psi} M^\dagger)\notag\\
    &=\cos^2\theta \tr_X\op{E_1}{E_1}+\sin^2\theta \tr_X\op{E_2}{E_2}\notag\\
    &+\cos\theta\sin\theta(e^{-i\phi}\tr_X\op{E_1}{E_2}+\text{h.c.})\notag\\
    \tr_X(M&U\op{\psi}{\psi}U^\dagger M^\dagger) \notag\\
    &= \cos^2\theta \tr_X\op{E_1}{E_1}+\sin^2\theta \tr_X\op{E_2}{E_2}\notag\\
    &+\cos\theta\sin\theta(e^{-i\phi}\lambda_1\lambda_2^* \tr_X\op{E_1}{E_2}+\text{h.c.})\notag.
\end{align}
    Since $M$ is a masker, the partial trace of both density matrices must be equal, which means that
    \begin{align}
        e^{-i\phi}(1-\lambda_1\lambda_2^*)\tr_X\op{E_1}{E_2}+\text{h.c.}=0.\notag
    \end{align}
Note that $(1-\lambda_1\lambda_2^*)\not=0$ since $\lambda_1\not=\lambda_2$.  For this to hold for any choice of $\phi$, we must have that 
\begin{align}
\label{Eq:partial-orthogonal-conds}
    \tr_X\op{E_1}{E_2}=\tr_X\op{E_2}{E_1}=0.
\end{align}
Let $\ket{E_1}=\sum_{i=1}^d\ket{i}_A\ket{\varphi_i}_B$ and $\ket{E_2}=\sum_{i=1}^d\ket{i}_A\ket{\varphi'_i}_B$ so that
\begin{align}       
0=\tr_B\op{E_1}{E_2}=\sum_{i,j=1}^d\op{i}{j}\ip{\varphi'_j}{\varphi_i},\notag
\end{align}
which requires that $\ip{\varphi_j'}{\varphi_i}=0$ for all $i,j$.  Therefore,
\begin{align}
\tr_A\op{E_1}{E_1}\tr_A\op{E_2}{E_2}&=\!\!\sum_{i,j=1}^d\ket{\varphi_i}\bra{\varphi_i}\ket{\varphi'_j}\bra{\varphi'_j}_B=0\notag.
\end{align}
This establishes the orthogonality of $\tr_A (M\op{e_1}{e_1}M^\dagger)$ and $\tr_A (M\op{e_2}{e_2}M^\dagger)$.  A similar argument shows orthogonality for Alice's reduced states.  It is interesting to note that Eq. \eqref{Eq:partial-orthogonal-conds}, the key step in this proof, also appears in the no-go state masking proof of Ref. \cite{Modi2018}. 
\end{proof}

We now use Proposition \ref{Prop:identity-unitary-mask} to prove our first main result, which is a full characterization of maskable gate sets.

\begin{theorem}
    \label{Thm:masking-gates}
    A set of $N$ unitary gates $\{U_n\}_{n=1}^N$ on $\mbb{C}^d$ can be masked iff $\{U_1^\dagger U_n\}_{n=2}^N$ forms a pairwise commuting set.
\end{theorem}
\begin{remark}
    Multiplying by $U_1^\dagger$ in this theorem is arbitrary since the set $\{U_1^\dagger U_n\}_{n=2}^N$ is pairwise commuting iff $\{U_k^\dagger U_n\}_{n=2}^N$ is pairwise commuting for any other $U_k$ in the set.  Indeed, the commuting relations $[U_1^\dagger U_m,U_1^\dagger U_k]=0$ for any $m,k$ implies that
    $U_k^\dagger U_m=U_1^\dagger U_m U_k^\dagger U_1$.  Inserting the identity $\mathbb{1} = U_1 U_1^\dag$ gives $U_k^\dagger U_1 U^\dagger_1 U_m=U_1^\dagger U_m U_k^\dagger U_1$, which says that $[U_k^\dag U_1, U_1^\dag U_m]=0$.  Then from the identity, 
\begin{align}
    \label{rmk1}
   & U_k^\dag U_m U_k^\dag U_n = (U_k^\dag U_1) (U_1^\dag U_m) (U_k^\dag U_1)(U_1^\dag U_n),
\end{align}
we use the commutator $[U_k^\dag U_1, U_1^\dag U_m]=0$ to freely swap the order as follows,
\begin{align*}
    (U_k^\dag U_1) (U_1^\dag U_m) &(U_k^\dag U_1)(U_1^\dag U_n)\\
   &= (U_k^\dag U_1) (U_1^\dag U_n) (U_k^\dag U_1)(U_1^\dag U_m)\\
   &= U_k^\dag U_n U_k^\dag U_m.
\end{align*}
Hence $[U_k^\dag U_m, U_k^\dag U_n]=0$.

\end{remark}

\noindent We now turn to the proof of Theorem \ref{Thm:masking-gates}.
\begin{proof}
    Let $\{U_n\}_{n=1}^N$ be an arbitrary set of unitaries on $\mbb{C}^d$.  By the observation made after Eq. \eqref{Eq:Channel-masking-defn} above, this is maskable iff the set $\{\mbb{1},W_n\}_{n=2}^N$ is also maskable, where $W_n:=U_1^\dagger U_n$.  Therefore, the proof of Theorem \ref{Thm:masking-gates} amounts to showing that the set $\{\mbb{1},W_n\}_{n=2}^N$ can be masked iff the $W_n$ are pairwise commuting.

($\Leftarrow$) Let $\{\mbb{1},W_n\}_{n=2}^N$ be a full set of pairwise commuting elements, and let $\{\ket{f_k}\}_{k=1}^d$ be a common orthonormal eigenbasis.  Then consider a bipartite masker $M:\mbb{C}^d\to \mbb{C}^d\otimes\mbb{C}^d$ of the form $M = \sum_{k=1}^{d} \ket{kk}\bra{f_k}$. An arbitrary state may be written in the eigenbasis as $\ket{\psi}=\sum_{k=1}^d \alpha_k\ket{f_k}\in\mbb{C}^d$.  Any unitary in $\{\mbb{1},W_n\}_{n=2}^N$ will act as $\ket{\psi}\mapsto \ket{\psi'}=\sum_{k=1}^d \alpha_ke^{i\phi_k}\ket{f_k}$ for some phases $\phi_k$.  The masker then transforms
\[M\ket{\psi'}=\sum_{k=1}^d\alpha_k e^{i\phi_k}\ket{kk},\]
whose reduced states are completely independent of the phases $\phi_k$ due to the orthonormality of the $\ket{k}$. The action of the arbitrary unitary in $\{\mbb{1},W_n\}_{n=2}^N$ has therefore been masked.

\noindent ($\Rightarrow$)  Suppose that a masker $M$ exists for the set $\{\mbb{1},W_n\}_{n=2}^N$.  We want to show that this is a pairwise commuting set.  An arbitrary pair $U,V\in\{W_n\}_{n=2}^N$ will commute with each other iff for every eigenvector $\ket{e}$ of $U$, the state $V\ket{e}$ is also an eigenvector of $U$ with the same eigenvalue as $\ket{e}$.  Now suppose on the contrary that $[U,V]\not=0$.  Then $U$ has some eigenvector $\ket{e}$ such that $V\ket{e}$ does not belong to the same eigenspace as $\ket{e}$.  This means there exists some other eigenvector $\ket{e'}$ of $U$ in a different eigenspace than $\ket{e}$ such that 
\begin{align}
    0\not=\bra{e'}V\ket{e}=\bra{e'}M^\dagger M V\ket{e}.
\end{align}
Since $M$ is a masker for the set $\{\mbb{1},V\}$, the bipartite states $MV\ket{e}$ and $M\ket{e}=:\ket{E}$ must purify the same reduced density matrices.  Hence, there exists a unitary $T$ such that $MV\ket{e}=\mbb{1}\otimes T\ket{E}$.  Denoting $\ket{E'}=M\ket{e'}$, the previous equation can be written as
\begin{align}
    0\not=\bra{E'}\mbb{1}\otimes T\ket{E}.
\end{align}
However, since $M$ also masks $\{\mathbb{1},U\}$, Proposition \ref{Prop:identity-unitary-mask} requires that $\tr_B\op{E'}{E'}\perp\tr_B\op{E}{E}$, which would imply that $\bra{E'}\mbb{1}\otimes T\ket{E}=0$.  From this contradiction, we conclude that $[U,V]=0$, which means the set  $\{\mbb{1},W_n\}_{n=2}^N$ is pairwise commuting.

\end{proof}

We note that our masker $M = \sum_{k=1}^d \op{kk}{f_k}$ acts by ``copying'' the basis index, i.e., $M\ket{f_k} = \ket{kk}$. Off-diagonal terms $\op{kk}{jj}$ vanish under partial trace for $k \neq j$ with this embedding. It is clear to see the off-diagonal correlations are thus globally accessible, but entirely hidden from the reduced channels.

%%%%%%%%%%%%%%%%%%%%%%%%%%%
%\todo{[Anna, please carefully check the above argument to make sure it works / makes sense...]}
%%%%%%%%%%%%%%%%%%%%%%%%%%%%%

%To prove necessity, consider the case of four qubit unitaries $\{U_1,U_2,U_3,U_4\}$ such that $\{U_1,U_2\}$ commute and $\{U_3,U_4\}$ commute separately. Assume the full set may be masked by a single masker. Let $W_i = U_1^\dag$. Then we consider the set $\{\mathbb{1}, W_2, W_3, W_4\}$. Now there exists a subset of three in which two operators do not commute. Without loss of generality, choose $\{\mathbb{1},W_2,W_3\}$. $W_2$ and $W_3$ do not commute, thus by Lemma 5 the set cannot be masked. Therefore the entire set $\{U_1,U_2,U_3,U_4\}$ cannot be masked by a single masker. Now generalize to $m$ sets of commuting qubit unitaries $\{U_n\}_m$, with a total of $N$ unitary operators. As long as no more than $N-2$ pairwise commute, the sets can always be restricted to two sets of separately commuting unitaries $\{U_1,U_2\}$ and $\{U_3,U_4\}$. Thus masking requires that $N-1$ unitiares pairwise commute. Therefore see any set of $N$ qubit unitaries may be masked iff $N-1$ of them commute pairwise, none of which is the identity. $\blacksquare$$

\begin{remark}
    Recall that every qubit pure state can be represented by a unit vector $\hat{n}$ on the surface of the unit sphere in $\mbb{R}^3$, and unitary gates $U\in SU(2)$ correspond to three-dimensional rotations $O\in SO(3)$. For qubit systems, an alternative and illustrative proof for the converse of Theorem \ref{Thm:masking-gates} can be given in terms of this Bloch sphere geometry.  Consider two states $\ket{r}, \ket{r'}$ that lie on the disk in $\mbb{R}^3$ perpendicular to the Bloch sphere axis of rotation $\hat{n}_r$ of $W_1$, and define $W_1 \ket{r} = \ket{r'}$. Denote the plane containing this disk by $\mc{P}$. Let $W_2$ map $\ket{r}, \ket{r'}$ to states $\ket{s}, \ket{s'}$ so that we have
\begin{align}
    W_1 \ket{r} &= \ket{r'},    &W_2 \ket{r} &= \ket{s},
    &W_2 \ket{r'} &= \ket{s'}.
\end{align}
If there exists a masker $M$ for $\{\mathbb{1},W_1,W_2\}$, then it must mask the set of states $\{\ket{r},\ket{r'},\ket{s}\}$, corresponding to the set of unitaries $\{\mathbb{1},W_1,W_2\}$ acting on $\ket{r}$.  Thus, $\tr_B[M\op{\psi}{\psi}M^\dag]$ must be constant for all $\ket{\psi} \in \{\ket{r},\ket{r'},\ket{s}\}$.  On the other hand, $M$ must also mask the set of states $\{\ket{r'},\ket{s'}\}$, corresponding to $\{\mathbb{1},W_2\}$ acting on $\ket{r'}$, which means that $\tr_B[M\op{\phi}{\phi}M^\dag]$ must also be constant for all $\ket{\phi} \in \{\ket{r'},\ket{s'}\}$.  Since $\ket{r'}$ is common to both sets $\{\ket{r},\ket{r'},\ket{s}\}$ and $\{\ket{r'},\ket{s'}\}$, the reduced density matrices must be constant when $M$ acts on all the $\{\ket{r},\ket{r'},\ket{s}, \ket{s'}\}$, i.e. $M$ masks this entire set of states.  However, it is known that this is possible iff the states lie on a single hyperdisk on the Bloch sphere \cite{Ding2020}, thus the set $\{\ket{r}, \ket{r'}, \ket{s}, \ket{s'}$ must all lie within the plane $\mc{P}$. By definition, $W_1$ is a rotation about the axis $\hat{n}_r$ perpendicular to $\mc{P}$. By Eq. (9), $W_2$ maps an arbitrary state $\ket{v} = a \ket{r} + b \ket{r'} \in \mc{P}$ to state $W_2 \ket{v} = a \ket{s} + b \ket{s'} \in \mc{P}$. Then we have $W_2 (\mc{P}) = \mc{P}$, i.e. the plane $\mc{P}$ is preserved under the rotation $W_2$. Therefore, the axis of rotation of $W_2$ must be $\hat{n}_r$, and so $W_1$ and $W_2$ share an axis of rotation and must commute. Remark 2 relies on the hyperdisk characterization of maskable qubit states \cite{Liang2019,Ding2020}. Maskable qudit sets for $d>2$, however, may extend beyond a single hyperdisk \cite{Ding2020}. Thus the hyperdisk argument does not generalize to higher dimensional systems, motivating our use of Theorem 1 for arbitrary dimensions.
\end{remark}

\begin{example}
    Consider the set of unitary operators composed of Pauli-X and Z gates $\{X,XZ,X \sqrt{Z}\}$. Note that this is not a pairwise commuting set, however multiplying by $X^\dag$ we consider the transformed set $\{\mathbb{1}, Z, \sqrt{Z}\}$. Now fully pairwise commuting, this set is maskable by Theorem 1, and thus the original set $\{X,XZ,X \sqrt{Z}\}$ is maskable. A masker $M$ may then be constructed by the common  eigenvectors of $\{\mathbb{1}, Z, \sqrt{Z}\}$, where the masker for the original isometry is defined by $M X^\dag$.  
\end{example}

\subsection{Masking noisy qubit channels}

We now move beyond the masking of gates and consider noisy channels.  Perhaps the simplest generalization involves mixing unitary gates with depolarizing noise.  These are families of channels $\{\mc{E}_{U_\lambda}^{(p)}\}_{\lambda}$ where
\begin{align}
    \mc{E}_{U_\lambda}^{(p)}(\cdot) &= p U_\lambda (\cdot) U_\lambda^\dag + (1-p) \frac{\mathbb{1}}{d}
\end{align}
with fixed $p$ and varying $U_\lambda$. One sees that  the $\{\mc{E}_{U_\lambda}^{(p)}\}_{\lambda}$ can be masked by masker $M$ iff for $X\in\{A,B\}$
\begin{align}
   p \tr_X[MU_\lambda(\cdot)U_\lambda^\dag M^\dag] &+ \frac{1}{d}(1-p) \tr_X[M\mathbb{1}M^\dag]
\end{align}
is a constant channel for all $U_\lambda$.  This holds iff the first term is independent of $\lambda$, and we therefore see that $\{\mc{E}_{U_\lambda}^{(p)}\}_{\lambda}$ can be masked for any $p>0$ iff the corresponding gates $\{U_\lambda\}_{\lambda}$ can be masked, for which we turn to Theorem \ref{Thm:masking-gates} to decide.  To investigate more interesting classes of channels, we restrict our attention to qubit systems.

\subsubsection{Pauli channels}

One of the most important types of qubit channels are the Pauli channels, which have the form
\begin{align}
    \mc{E}_{\vec{p}}(\cdot) = \sum_i p_i \sigma_i (\cdot) \sigma_i^\dag
\end{align}
for $i = \{0,x,y,z\}$ and probability four-vector $\vec{p}=(p_0,p_x,p_y,p_z)$. We define $\mathfrak{M}$ as an arbitrary set of probability four-vectors. Masking of channels $\{\mc{E}_{\vec{p}}\}_{\mathfrak{M}}$ for arbitrary input state $\rho$ requires that $\tr_X(M \mc{E}_{\vec{p}}(\rho) M^\dag)$ is fixed for all $\vec{p}\in\mathfrak{M}$, with $X,Y \in \{A,B\}$. 
\begin{theorem}
    Let $\mathfrak{M}$ denote an arbitrary set of probability four-vectors.  A family of Pauli channels $\{\mc{E}_{\vec{p}}\}_{\vec{p}\in \mathfrak{M}}$ can be masked iff there exists some $k\in\{x,y,z\}$ and constant $c\in[0,1]$ such that $p_{0} + p_{k}=c$ for all $\vec{p}\in\mathfrak{M}$.
    \label{Thm:masking-Paulis}
\end{theorem}
\begin{proof}
($\Leftarrow$)  Suppose without loss of generality that $k=x$ so that $p_0+p_x=c$.  Define the masker $M:\mbb{C}^2\to\mbb{C}^2\otimes\mbb{C}^2$ by $M=\op{00}{+}+\op{11}{-}$, where $\ket{\pm}$ are the $\pm 1$ eigenstates of $\sigma_x$.  Then a straightforward calculation shows that for any input state $\rho$ we have
\begin{align}
    \tr_X[M&\mc{E}_{\vec{p}}(\rho)M^\dagger]\notag\\
    =&(p_0+p_x)(\bra{+}\rho\ket{+}\op{0}{0}+\bra{-}\rho\ket{-}\op{1}{1})\notag\\
    +&(p_y+p_z)(\bra{-}\rho\ket{-}\op{0}{0}+\bra{+}\rho\ket{+}\op{1}{1}),\notag
\end{align}
which is the same for all $\vec{p}\in\mathfrak{M}$.

\noindent ($\Rightarrow$)  Suppose that a masker $M$ exists for the channels $\{\mc{E}_{\vec{p}}\}_{\vec{p}\in \mathfrak{M}}$.  Let us denote action of $M$ on the computational basis as $M\ket{0}=\ket{\Psi_0}_{AB}$ and $M\ket{1}=\ket{\Psi_1}_{AB}$.  Then
\begin{align}
    &\tr_X[M\mc{E}_{\vec{p}}(\op{0}{0}) M^\dagger]\notag\\
    &=(p_0+p_z)\tr_{X}\op{\Psi_0}{\Psi_0}+(p_x+p_y)\tr_X\op{\Psi_1}{\Psi_1}\notag.
\end{align}
If there exists two vectors $\vec{p},\vec{p}^{\;'} \in\mathfrak{M}$ such that $p_0+p_z\not=p_0'+p_z'$. Let  $p_0 + p_z = p_0' + p_z' + c$, then $p_x + py = p_x' + p_y' - c$. To satisfy the masking condition $\tr_X[M\mc{E}_{\vec{p}}(\op{0}{0}) M^\dagger]=\tr_X[M\mc{E}_{\vec{p}'}(\op{0}{0}) M^\dagger]$, we have
\begin{align*}
    &(p_0' + p_z' + c) \tr_X\op{\Psi_0}{\Psi_0} + (p_x' + p_y' - c) \tr_X\op{\Psi_1}{\Psi_1} \\
    &= (p_0' + p_z') \tr_X\op{\Psi_0}{\Psi_0} + (p_x' + p_y') \tr_X\op{\Psi_1}{\Psi_1}
\end{align*}
which requires that $\tr_X\op{\Psi_0}{\Psi_0}=\tr_X\op{\Psi_1}{\Psi_1}$.  Under this assumption, one then computes
\begin{align}
    &\tr_X[M\mc{E}_{\vec{p}}(\op{+}{+}) M^\dagger]\notag\\
    &=\tr_X\op{\Psi_0}{\Psi_0}+(p_0+p_x-\tfrac{1}{2})(\tr_X\op{\Psi_0}{\Psi_1}+\text{h.c.)},\notag\\
    &\tr_X[M\mc{E}_{\vec{p}}(\op{\widetilde{+}}{\widetilde{+}}) M^\dagger]\notag\\
    &=\tr_X\op{\Psi_0}{\Psi_0}-i(p_0+p_y-\tfrac{1}{2})(\tr_X\op{\Psi_0}{\Psi_1}-\text{h.c.)},
\end{align}

where $\ket{\widetilde{+}}=\frac{1}{\sqrt{2}}(\ket{0}+i\ket{1})$.  If there exists two vectors $\vec{p},\vec{p}^{\;'}\in\mathfrak{M}$ such that $p_0+p_x\not=p_0'+p_x'$, then combining the masking assumption $\tr_X[M\mc{E}_{\vec{p}}(\op{0}{0}) M^\dagger]=\tr_X[M\mc{E}_{\vec{p}'}(\op{0}{0}) M^\dagger]$ with the first equality above gives $\tr_X\op{\Psi_0}{\Psi_1}=-\tr_X\op{\Psi_1}{\Psi_0}$.  Under this further assumption, if there exists two vectors $\vec{p},\vec{p}^{\;'}\in\mathfrak{M}$ such that $p_0+p_y\not=p_0'+p_y'$, then the second equality above would imply that $\tr_X\op{\Psi_0}{\Psi_1}=\tr_X\op{\Psi_1}{\Psi_0}$, which would mean that $\tr_X\op{\Psi_1}{\Psi_0}=0$.  Yet, it is not possible that both $\tr_X\op{\Psi_0}{\Psi_0}=\tr_X\op{\Psi_1}{\Psi_1}$ and $\tr_X\op{\Psi_1}{\Psi_0}=0$. 

Recall that $\tr_X\op{\Psi_0}{\Psi_0}=\tr_X\op{\Psi_1}{\Psi_1}$ implies that $\ket{\Psi_0}$ and $\ket{\Psi_1}$ are two purifications of the same state. Then we have the relation $\ket{\Psi_1} = (U_A \otimes \mathbb{1}_B) \ket{\Psi_0}$. Write $\ket{\Psi_0}$ in Schmidt form
\begin{align*}
    \ket{\Psi_0} &= \sqrt{\lambda} \ket{a_0}_A \ket{b_0}_B + \sqrt{1-\lambda} \ket{a_1}_A \ket{b_1}_B\\
\end{align*}
for some orthonormal bases $\{a_i\}$ and $\{b_i\}$.
When taking the partial trace, cross terms will be eliminated, and we find
\begin{equation}
    \tr_B \op{\Psi_1}{\Psi_0} = \begin{pmatrix}
        \lambda u_{00} & (1-\lambda) u_{01} \\
        \lambda u_{10} & (1-\lambda) u_{11}
    \end{pmatrix}
    =0
\end{equation}
If $0<\lambda<1$, i.e. $\rho_A$ is rank 2, this forces all $u_{ij}=0$, which is impossible since $U_A$ is unitary. In the case that $\rho_A$ is a pure state, $\ket{\Psi_0}$ must be a product state. Define $\ket{\Psi_0} = \ket{a}_A \otimes \ket{b}_B$. Then $\ket{\Psi_1}$ must satisfy $\ket{\Psi_1} = U_A \ket{a}_A \otimes \ket{b}_B$ for some unitary operator $U_A$, but this implies $\tr_B\op{\Psi_1}{\Psi_0} \neq 0$.

Hence, we have reached a contradiction, and one of our three passing assumptions cannot be true.  That is, there must exist some $k\in\{x,y,z\}$ such that $p_{0}+p_k$ is constant for all $\vec{p}\in\mathfrak{M}$.
    
\end{proof}

\begin{example}

    \noindent The equality $p_0+p_k=c$ in Theorem \ref{Thm:masking-Paulis} along with normalization imposes two linear constraints on the probability vectors $\vec{p}$.  Almost any collection of maskable Pauli channels belongs to a two-parameter family of maskable channels (the exception being the case when $p_0+p_k=1$).  For example, when $k=X$ and $c<1$, we have a two-parameter family of maskable channels given by
    \begin{align}
           \mc{E}_{\mu,\nu}(\cdot)=&\mu (\cdot)+(c-\mu)\sigma_X(\cdot)\sigma_X\notag\\
           &+\nu \sigma_Y(\cdot)\sigma_Y+(1-c-\nu)\sigma_Z(\cdot)\sigma_Z\
    \end{align}
    for $\mu\leq c$ and $\nu\leq 1-c$.  This example shows a notable difference between channel and state masking.  Namely, every collection of maskable qubit states belongs to a one-parameter family of states corresponding to a circle on the Bloch sphere. In contrast, channel masking allows for multi-parameter families of maskable objects.
   
\end{example}

\subsubsection{Any family containing the identity}

We now turn to the special problem of masking a channel against the identity, $\id$.  As described in the introduction, this has the appealing interpretation of pushing all the noise of a given channel into the correlations between two subsystems, while leaving the reduced state dynamics unaffected.  Here, we completely characterize all the qubit channels that allow for such a process. We seek solutions to the equation
\begin{align}
\label{Eq:identity-masking-canonical}
\tr_X [M\rho M^\dagger] = \tr_X[M\mc{E}(\rho) M^\dagger],\quad X\in\{A,B\},\;\forall \rho.
\end{align}
A full solution to this problem for the case of qubit channels is given in the following.

\begin{theorem}
    \label{Thm:masking-identity}
    The set of qubit channels $\{\id,\mc{E}\}$ can be masked iff $\mc{E}$ is unital and has a pure state fixed point; i.e. $\mc{E}(\op{\psi}{\psi})=\op{\psi}{\psi}$ for some $\ket{\psi}$.
\end{theorem}
\begin{remark}
    Geometrically, we can understand these channels as being maps on the Bloch sphere that preserve the origin as well as two anti-podal points on the surface of the Bloch sphere.  The fact that $\mc{E}$ is unital with $\ket{\psi}$ as a fixed point means that $\ket{\psi^\perp}$ is also a fixed point, where $\ket{\psi^\perp}$ is orthogonal to $\ket{\psi}$ and satisfying $\mbb{1}=\op{\psi}{\psi}+\op{\psi^\perp}{\psi^\perp}$.
\end{remark}

\begin{remark}
    
While Theorem \ref{Thm:masking-identity} involves masking just a single qubit channel $\mc{E}$, it can easily be extended to include more channels.  Suppose that $\{\id,\mc{E}_\lambda\}_\lambda$ is a family qubit channels that can be masked.  Then so can the pair of channels $\{\id,\sum_\lambda p_\lambda \mc{E}_\lambda\}$, where $p_\lambda$ is a probability distribution with $p_\lambda>0$ and chosen such that $\sum_\lambda p_\lambda \mc{E}_\lambda(\mbb{1})\not=\mbb{1}$ if one of the channels $\mc{E}_\lambda$ is non-unital (and arbitrarily chosen otherwise) \footnote{To show that such a $p_\lambda$ can always be chosen this way, suppose $\mc{E}_1(\mbb{I})=\sigma\not=\mbb{I}$ and $p_1\mc{E}_1(\mbb{I})+\sum_{i>1}p_i\mc{E}_i(\mbb{I})=\mbb{I}$.  Then by considering a sufficiently small perturbation $p_1\mapsto(1-\epsilon)p_1$ and $p_i\mapsto(1+\epsilon\frac{p_1}{1-p_1})p_i$ for $i>1$, we have $(1-\epsilon)p_1\mc{E}_1(\mbb{I})+(1+\epsilon\frac{p_1}{1-p_1})\sum_{i>1}p_i\mc{E}_i(\mbb{I})=\mbb{I}-\epsilon\frac{p_1}{1-p_1}(\sigma-\mbb{I})$, which does not equal the identity for all $\epsilon>0$}. By Theorem \ref{Thm:masking-identity}, the convex combination $\sum_\lambda p_\lambda \mc{E}_\lambda$ must have a pure-state fixed point.  This is possible only if each of the individual channels $\mc{E}_\lambda$ have the same pure-state fixed point.  Moreover, $\sum_\lambda p_\lambda \mc{E}_\lambda$ must be unital, which is not possible by our choice of $p_\lambda$ unless all the $\mc{E}_\lambda$ are unital themselves.  We therefore conclude that the $\mc{E}_\lambda$ must all be unital and have a common pure-state fixed point.  Conversely, if all the $\mc{E}_\lambda$ are unital and have the same pure-state fixed point, then the masker constructed in Lemma \ref{Lem:unital-masking} below will be a masker for the set $\{\id,\mc{E}_\lambda\}_\lambda$, for the same reason as given in that proof.  We summarize in the following.
\begin{corollary}
A family of qubit channels $\{\id,\mc{E}_\lambda\}_\lambda$ can be masked iff all the $\mc{E}_\lambda$ are unital and possess a common pure-state fixed point.
\end{corollary}

\end{remark}

To prove Theorem \ref{Thm:masking-identity}, we continue with the geometrical picture and recall in more detail that every qubit channel $\mc{E}$ can be represented by an affine transformation on the Bloch vectors $\mbf{n} \rightarrow A \mbf{n}+\mbf{b}\in\mbb{R}^3$ \cite{Nielsen-Chuang-2010a}.  When the channel is unital (meaning that $\mc{E}(\mbb{1})=\mbb{1}$), the vector $\mbf{b}$ vanishes. The unital case is considered first.
\begin{lemma}
\label{Lem:unital-masking}
    If $\mc{E}$ is a unital qubit channel, then the set $\id,\mc{E}\}$ is maskable iff $\mc{E}$ has a pure state fixed point; i.e. $\mc{E}(\op{\psi}{\psi})=\op{\psi}{\psi}$ for some $\ket{\psi}$.
\end{lemma}
\begin{proof}
We begin by noting that a unital channel $\mc{E}$ with Bloch sphere action $\mbf{n}\mapsto A\mbf{n}$ has a pure-state fixed point iff $A$ has an eigenvalue of $\lambda=1$.   ($\Rightarrow$) Now suppose that $\{\mathbb{1},\mc{E}\}$ is maskable.  Let the Bloch sphere action of $\mc{E}$ be given by $\mbf{n} \rightarrow A \mbf{n}$.  We claim that $A$ must have an eigenvalue of $\lambda=1$.  Suppose on the contrary that it does not.  Then $A-\mbb{1}$ is non-singular, and so for any vector $\mbf{v}$, we can identify the nonzero vector $\mbf{w}=(A-\mbb{1})^{-1}\mbf{v}$ such that $A\hat{\mbf{w}}=\hat{\mbf{w}}+\mbf{v}/\Vert\mbf{w}\Vert$, where $\hat{\mbf{w}}=\mbf{w}/\Vert\mbf{w}\Vert$.  Then
\begin{subequations}
\begin{align}
    \tr_X[M&\mc{E}(\op{\hat{\mbf{w}}}{\hat{\mbf{w}}})M^\dagger]=\tr_X[M\frac{1}{2}(\mbb{1}+\frac{\mbf{w}+\mbf{v}}{\Vert\mbf{w}\Vert}\cdot\vec{\sigma})M^\dagger],\notag\\
    \tr_X[M&\op{\hat{\mbf{w}}}{\hat{\mbf{w}}}M^\dagger]=\tr_X[M\frac{1}{2}(\mbb{1}+\frac{\mbf{w}}{\Vert\mbf{w}\Vert}\cdot\vec{\sigma})M^\dagger].\notag
\end{align}
\end{subequations}
The masking conditions of Eq. \eqref{Eq:identity-masking-canonical} then implies that $\tr_X[M(\mbf{v}\cdot\vec{\sigma})M^\dagger]=0$, which further means that
\begin{align}
    \tr_X[M\op{\hat{\mbf{v}}}{\hat{\mbf{v}}}M^\dagger]&=\tr_X[M\frac{1}{2}(\mbb{1}+\hat{\mbf{v}}\cdot\vec{\sigma})M^\dagger]\notag\\
    &=\frac{1}{2}\tr_X[MM^\dagger].
\end{align}
Since $\hat{\mbf{v}} = \mbf{v}/ \Vert\mbf{v}\Vert$ is arbitrary, and the RHS is independent of $\mbf{v}$, we see that $M$ would need to be a masker for all qubit states, which is impossible \cite{Modi2018}.  We therefore conclude that $A$ must have an eigenvalue of $\lambda=1$, and so $\mc{E}$ has a pure-state fixed point.

($\Leftarrow$)  Suppose that $\mc{E}(\op{\psi}{\psi})=\op{\psi}{\psi}$.  Let $U$ be a unitary such that $U\ket{0}=\ket{\psi}$ such that the conjugated channel $\mc{E}':=\mc{U}^\dagger \circ\mc{E}\circ\mc{U}$ has $\ket{0}$ as a fixed point.  Since   $\{\id,\mc{E}\}$ is maskable iff $\{\id,\mc{E}'\}$ is maskable, it suffices to prove that the latter set is maskable.  

With $\op{0}{0}=\frac{1}{2}(\mbb{1}+\sigma_z)$ being a fixed point of $\mc{E}'$, the Bloch sphere transformation matrix of $\mc{E}'$ will satisfy $A\hat{z}=\hat{z}$ (and so $\mc{E}'(\sigma_z)=\sigma_z$).  Furthermore, since $\Vert A\Vert\leq 1$, the matrix $A$ must have block form
\[A=\begin{pmatrix}a&b&0\\c&d&0\\0&0&1\end{pmatrix}.\]
We therefore see that the channel $\mc{E}'$ acts invariantly on the $x-y$ plane of the Bloch sphere.  With this observation, we can prove that the simple isometry $M\ket{0}=\ket{00}$ and $M\ket{1}=\ket{11}$ serves as a masker for $\{\id,\mc{E}'\}$.  To see this, note that
\begin{align}
    &\tr_X[MM^\dagger]=\mbb{1},\;\;\tr_X[M\sigma_zM^\dagger]=\sigma_z,\notag\\
    &\tr_X[M\sigma_xM^\dagger]=\tr_X[M\sigma_yM^\dagger]=0.\label{Eq:partial-Pauli-vanish}
\end{align}
Then for an arbitrary state $\rho=\frac{1}{2}(\mbb{1}+\mbf{n}\cdot\vec{\sigma})$, we have
\begin{align}
    \tr_X[M\mc{E}'(\rho)M^\dagger]&=\frac{1}{2}\tr_X[M\mc{E}'(\mbb{1}+\mbf{n}\cdot\vec{\sigma})M^\dagger]\notag\\
    &=\frac{1}{2}\tr_X[MM^\dagger+n_zM\sigma_z M^\dagger\notag\\
    &\qquad\quad\quad+M\mc{E}(n_x\sigma_x+n_y\sigma_y)M^\dagger]\notag\\
    &=\frac{1}{2}(\mbb{1}+n_z\sigma_z),
    \label{Eq:unital-masker-1}
\end{align}
where the second equality follows from the facts that $\mc{E}'$ is unital and $\mc{E}'(\sigma_z)=\sigma_z$, while the third equality follows from the facts that $\mc{E}'$ acts invariantly on the $x-y$ plane and Eq. \eqref{Eq:partial-Pauli-vanish}.  At the same time,
\begin{align}
    \tr_X[M\rho M^\dagger]&=\frac{1}{2}\tr_X[MM^\dagger+n_zM\sigma_z M^\dagger\notag\\
    &\qquad\quad\quad+M(n_x\sigma_x+n_y\sigma_y)M^\dagger]\notag\\
    &=\frac{1}{2}(\mbb{1}+n_z\sigma_z),
    \label{Eq:unital-masker-2}
\end{align}
where we again use Eq. \eqref{Eq:partial-Pauli-vanish}.  Since $\rho$ was arbitrary, a comparison of Eqns. \eqref{Eq:unital-masker-1} and \eqref{Eq:unital-masker-2} shows that $M$ is a masker for $\{\id,\mc{E}'\}$.
\end{proof}

To complete the proof of Theorem \ref{Thm:masking-identity}, we now consider the non-unital case.
\begin{lemma}
    If $\mc{E}$ is a non-unital qubit channel, then it is not possible to mask the set $\{\id,\mc{E}\}$.
\end{lemma}

\begin{proof}

    Let $\mc{E}$ be an arbitrary non-unital channel, and suppose on the contrary that a masker exists for $\{id,\mc{E}\}$.  The action of $\mc{E}$ on the Bloch sphere is now described by an affine transformation $\mbf{n}\mapsto A\mbf{n}+\mbf{b}$, with $\mbf{b}\not=0$.  This immediately implies that $\mbb{I}-A$ is invertible.  Indeed, otherwise $A$ would have an eigenvalue of $\lambda=1$, meaning that $\pm\hat{\mbf{n}}\mapsto \pm\hat{\mbf{n}} +\mbf{b}$.  However, one of the vectors $\pm\hat{\mbf{n}} +\mbf{b}$ will have norm larger than one, which is not possible for a CPTP map.
    
    Now for an arbitrary input Bloch vector $\hat{\mbf{n}}$, the masking condition of Eq. \eqref{Eq:identity-masking-canonical} can then be stated in the Pauli representation as
    \begin{align}
        &\tr_X[M(\mathbb{I}+\hat{\mbf{n}}\cdot\vec{\sigma})M^\dagger]=\tr_{X}[M(\mbb{I}+(A\hat{\mbf{n}}+\mbf{{b}})\cdot\vec{\sigma})M^\dagger]\notag\\
        \Rightarrow \quad&\tr_X[M(\hat{\mbf{n}}-A\hat{\mbf{n}})\cdot\vec{\sigma}M^\dagger]=\tr_X[M\mbf{b}\cdot\vec{\sigma}M^\dagger]
    \end{align}
Therefore, we have that
\begin{align}
\label{Eq:constant-n}
    \tr_X[M(\mbb{I}-A)\hat{\mbf{n}}\cdot\vec{\sigma}M^\dagger]
\end{align}
is constant for all unit vectors $\hat{\mbf{n}}$.  Since $\mathbb{I}-A$ is non-singular, it will have singular value decomposition $\mbb{I}-A=\sum_{i=1}^3 \sigma_i\mbf{\hat{w}}_i\cdot\hat{\mbf{v}}^T_i$, where $\sigma_i>0$, and both sets of vectors $\{\hat{\mbf{v}}_i\}_{i=1}^3$ and $\{\hat{\mbf{w}}_i\}_{i=1}^3$ are orthonormal.  Consider the eight unit vectors labeled by three bits $(b_1,b_2,b_3)\in\{0,1\}^{\times 3}$:
\begin{equation}
    \hat{\mbf{s}}_{b_1,b_2,b_3}:=\frac{1}{\sqrt{3}}\left[(-1)^{b_1}\hat{\mbf{v}}_1+(-1)^{b_2}\hat{\mbf{v}}_2+(-1)^{b_3}\hat{\mbf{v}}_3\right].
\end{equation}
Under the action of $\mbb{I}-A$, we have the transformation
\begin{align}
   \mbf{t}_{b_1,b_2,b_3}:= (\mbb{I}-A)\hat{\mbf{s}}_{b_1,b_2,b_3}=\frac{1}{\sqrt{3}}\sum_{i=1}^3(-1)^{b_i}\sigma_i\hat{\mbf{w}}_i.
\end{align}
While the $\mbf{t}_{b_1,b_2,b_3}$ are no longer unit vectors, they will all have the same norm $\Vert \mbf{t}_{b_1,b_2,b_3}\Vert$ due to the orthogonality of the $\hat{\mbf{w}}_i$.  Hence, from Eq. \eqref{Eq:constant-n}, we conclude that
\begin{align}
    \frac{\tr_X[M(\mbb{I}-A)\hat{\mbf{s}}_{b_1,b_2,b_3}\cdot\vec{\sigma}M^\dagger]}{\Vert \mbf{t}_{b_1,b_2,b_3}\Vert}&=\frac{\tr_X[M \mbf{t}_{b_1,b_2,b_3}\cdot\vec{\sigma}M^\dagger]}{\Vert \mbf{t}_{b_1,b_2,b_3}\Vert}\notag\\
    &=\tr_X[M \hat{\mbf{t}}_{b_1,b_2,b_3}\cdot\vec{\sigma}M^\dagger]
\end{align}
are constant for all $(b_1,b_2,b_3)\in\{0,1\}^{\times 3}$, where now the $\hat{\mbf{t}}_{b_1,b_2,b_3}$ are unit vectors.  Consequently,
\begin{equation}
    \tr_X[M \frac{1}{2}(\mbb{I}+\hat{\mbf{t}}_{b_1,b_2,b_3}\cdot\vec{\sigma})M^\dagger]=\tr[M\op{\hat{\mbf{t}}_{b_1,b_2,b_3}}{\hat{\mbf{t}}_{b_1,b_2,b_3}}M^\dagger]
\end{equation}
is also constant for all $(b_1,b_2,b_3)\in\{0,1\}^{\times 3}$.  This says we have a masker for the states $\ket{\hat{\mbf{t}}_{b_1,b_2,b_3}}$.  But this is impossible since the eight vectors $\hat{\mbf{t}}_{b_1,b_2,b_3}$ lie in eight different octants of the Bloch sphere, and therefore not on the same hyperdisk.  Indeed, every hyperplane in $\mbb{R}^3$ will fail to intersect at least one octant \cite{yao1983}. 

This can be seen via the equation of a plane $P$ in $\mbb{R}^3$
\begin{align*}
    Ax + By + Cz + D=0
\end{align*}
where $(A,B,C) \neq (0,0,0)$.
Without loss of generality, assume $A,B,C > 0$, as a change in sign will simply reflect $P$ across the corresponding axis, permuting the octants intersecting $P$. Suppose that $P$ intersects all eight octants. Then $P$ must contain a point with $x,y,z > 0$ corresponding to the $(+,+,+)$ octant, and a point with $x,y,z < 0$ in the $(-,-,-)$ octant.

If $D > 0$, then every $(x,y,z)\in P$ satisfies $Ax + By + Cz < 0 $, but for points $x,y,z >0$ with $A,B,C > 0$, we have $Ax + By + Cz > 0$. Therefore $P$ cannot contain any points in the $(+,+,+)$ octant. If $D < 0$, every $(x,y,z) \in P$ must satisfy $Ax + By + Cz > 0$, but this is impossible for $x,y,z <0$ and $A,B,C >0$. Thus $P$ cannot contain any points in the $(-,-,-)$ octant. If $D=0$, $Ax + By + Cz = 0$ cannot be satisfied by any point with $x,y,z>0$ or $x,y,z<0$, and $P$ cannot intersect octants $(+,+,+)$ and $(-,-,-)$. Therefore, $P$ can intersect at most 7 octants.
\end{proof}

\subsection{Masking classical channels}

We close our study by considering the classical analog of channel masking.  A classical channel on discrete sets $\mc{X}\to \mc{Y}$ is given by a collection of conditional probability distributions $p(y|x)$, with $x\in\mc{X}$ and  $y\in\mc{Y}$.  The corresponding CPTP then has the form $\mc{C}(\cdot)=\sum_{x,y}p(y|x)\op{y}{y}\bra{x}\cdot\ket{x}$.  Moreover, the classical analog of a unitary channel is just a permutation $\Pi$.

We first observe that it is impossible for any classical circuit to mask two distinct permutations $\Pi_1$ and $\Pi_2$ on some set $\mc{X}=\{1,2,\cdots,d\}$.  Let $\ket{x}$ be such that $\Pi_1\ket{x}=\ket{y}$ and $\Pi_2\ket{x}=\ket{y'}$ with $y\not=y'$.  Any reversible classical masker for $\{\Pi_1,\Pi_2\}$ would be a one-to-one mapping $M:\{1,\cdots,d\}\mapsto \{1,\cdots,d\}^{\times 2}$.  However, this means that $M\ket{y}=\ket{z_1}\ket{z_2}$ and $M\ket{y'}=\ket{z_1'}\ket{z_2'}$ for $\{z_1,z_2\}\not=\{z_1',z_2'\}$.  Therefore, the final state for at least one of the subsystems will differ depending on whether $\Pi_1$ or $\Pi_2$ is applied to input $\ket{x}$; i.e. masking is not possible.

On the other hand, using a quantum masker \textit{any} family of classical channels can be masked.
\begin{lemma}
    Any family of classical channels with finite input and output sets $\mc{X}$ and $\mc{Y}$ can be masked by a quantum masker.
\end{lemma}

\begin{proof}
    Let $S = \{\mc{C}_\lambda\}_\lambda$ be a collection of classical channels having the form
\begin{equation}
    \mc{C}_\lambda(\cdot) = \sum_{x\in\mc{X},y\in\mc{Y}} p_\lambda (y|x) \braket{x| \cdot | x} \op{y}{y}.
\end{equation}
Without loss of generality, suppose $\mc{Y}=\{1,2,\cdots,d\}$. We then define a quantum masker $M :\mbb{C}^d\to \mbb{C}^d\otimes\mbb{C}^d$ for $\mc{S}$ by its action:
\begin{align}
    M\ket{j} = \frac{1}{\sqrt{d}}\sum_{k=1}^{d}w^{kj}\ket{kk}
\end{align}
where $w = e^{2\pi i /d}$. One can easily verify that the reduced density matrices on subsystems $A$ and $B$ are always maximally mixed, $\mbb{1}_d / d$.
\end{proof}
\vspace{0.1mm}

\section{Conclusion}
In this paper, we have introduced and developed the concept of quantum channel masking, a dynamical extension of state masking in which the identity of a channel is hidden from local subsystems but remains globally recoverable. We provided a complete characterization of maskable sets of qudit unitary gates $\{U_n\}_{n=1}^N$, showing that masking is possible if and only if $\{U_1^\dagger U_n\}_{n=2}^N$ forms a pairwise commuting set. For the special case of qubit unitaries, we translated this into a geometric picture on the Bloch sphere corresponding to rotational symmetries.

We extended our analysis of channel masking to noisy qubit channels, characterizing all families of maskable Pauli channels.  Interestingly, there exist two-parameter families of Pauli channels that can be masked. %We further prove that masking holds under depolarizing noise, provided the underlying unitary structure of the channels satisfies the commutativity condition of our main theorem. In the case of qubit channels, 
We further showed that non-unital channels cannot be masked against the identity, whereas certain families of unital qubit channels allow for such masking.  This can be interpreted as a scenario in which the effect of noise is completely hidden from each subsystem. Beyond quantum systems, we showed that while classical channels cannot be masked by classical circuits, quantum maskers can successfully mask any set of arbitrary classical channels. Thus we have demonstrated a clear operational advantage of quantum operations.

Our work establishes a foundation for understanding concealment of quantum operations themselves, with implications in secure quantum information processing including applications in quantum secret sharing \cite{Karlsson1999, Hillery1999,Zukowski1998}, error correction \cite{Li2018,Han-2020a}, and bit commitment \cite{Mayers1997,Modi2018}. In particular, our results show that a secret can be encoded in a quantum channel and distributed among multiple parties such that no single party may recover it. Thereby our framework suggests a natural channel-based secret-sharing primitive by encoding a secret into the channel label. From the perspective of error correction, we have demonstrated that the effect of noise can be hidden entirely in the correlations of subsystems in the masking of $\{\id, \mc{E}\}$. This leaves reduced systems entirely unaffected by noise.

Several promising directions for future research emerge from our findings. A natural extension is to investigate channel masking mixed ancilla states, as we consider only the pure ancilla case. This may broaden the class of maskable channels. Additionally, Modi et al. propose a no-qubit commitment protocol based on state masking, in which the committing party can always cheat \cite{Modi2018}. It would be of interest to explore the channel masking analogue to this qubit commitment protocol. Finally, extending geometrical interpretations of channel masking beyond qubits to higher-dimensional systems may reveal interesting structural limitations.\\

\section*{Acknowledgments}
This work was supported by NSF Grant No. 2112890 and the IBM-Illinois Discovery Accelerator Institute.

%\todo{Is there analogous notion of channel quantum bit commitment for our problem?}  

%Here we have largely considered qubit states. Although we present a construction for a masker of $N$ qudit channels such that $N-1$ commute, this condition is sufficient but has not been proven to be necessary for d-dimensional qudit unitaries with degenerate eigenvalues. The idea of masking states on qudit states on higher-dimensional hyperdisks has been explored in ref. \cite{Ding2020}. We conjecture the necessity of higher-dimensional hyperdisks for masking of general sets of qudit channels involving degeneracy. This remains to be proven in future work, as well as uniqueness of the corresponding masker.

\nocite{*}
\bibliography{paper}

\end{document}